\documentclass[prl,aps,superscriptaddress,nofootinbib,twocolumn]{revtex4-2}
\newcommand{\abs}[1]{\vert #1 \vert}

\usepackage{amsthm,amsmath,amssymb, physics, mathtools}
\usepackage{mathrsfs}

\usepackage{amsthm}
\newtheorem{theorem}{Theorem}

\theoremstyle{definition}

\theoremstyle{remark}

\usepackage{comment}
\usepackage{amssymb}
\usepackage{graphicx}
\usepackage{epstopdf}

\usepackage{algorithm}
\usepackage[noend]{algpseudocode}

\usepackage{natbib}

\usepackage[colorlinks=true,citecolor=blue,linkcolor=magenta]{hyperref}

\usepackage{xcolor}

\usepackage{comment}

\usepackage[normalem]{ulem}
\usepackage{xcolor}
\newcommand\redsout{\bgroup\markoverwith{\textcolor{red}{\rule[0.5ex]{2pt}{0.4pt}}}\ULon}

\newcommand{\ph}{p^{\mathrm{h}}}
\newcommand{\pc}{p^{\mathrm{c}}}
\newcommand{\pth}{\tilde{p}^{\mathrm{h}}}
\newcommand{\ptc}{\tilde{p}^{\mathrm{c}}}
\newcommand{\psih}{\Psi^{\mathrm{h}}}
\newcommand{\psic}{\Psi^{\mathrm{c}}}
\newcommand{\psith}{\tilde{\Psi}^{\mathrm{h}}}

\newcommand{\rh}{r^{\mathrm{h}}}
\newcommand{\sh}{s^{\mathrm{h}}}
\newcommand{\rc}{r^{\mathrm{c}}}

\newcommand{\fn}{f}
\newcommand{\FN}{D_{f}}
\newcommand{\FI}{D_{I}}
\newcommand{\FNh}[1]{\FN^{\mathrm{h}}(#1)}
\newcommand{\FNc}[1]{\FN^{\mathrm{c}}(#1)}
\newcommand{\FNx}[1]{\FN^{\mathrm{x}}(#1)}

\newcommand{\FIh}[1]{\FI^{\mathrm{h}}(#1)}
\newcommand{\FIc}[1]{\FI^{\mathrm{c}}(#1)}

\begin{document}

\title{Theory of Quantum Imaginary-Time Mpemba Effect}
\author{Yumeng Zeng}
\thanks{These authors contributed equally to this work.}
\affiliation{Center on Frontiers of Computing Studies, School of Computer Science, Peking University, Beijing 100871, China}

\author{Jeongrak Son}
\thanks{These authors contributed equally to this work.}
\affiliation{Nanyang Quantum Hub, School of Physical and Mathematical Sciences, Nanyang Technological University, 637371, Singapore}

\author{Mile Gu}
\email{mgu@quantumcomplexity.org}
\affiliation{Nanyang Quantum Hub, School of Physical and Mathematical Sciences, Nanyang Technological University, 637371, Singapore}
\affiliation{Center for Quantum Technologies, Nanyang Technological University, 639798, Singapore}
\affiliation{MajuLab, CNRS-UNS-NUS-NTU International Joint Research Unit, UMI 3654, Singapore 117543, Singapore}

\author{Xiao Yuan}
\email{xiaoyuan@pku.edu.cn}
\affiliation{Center on Frontiers of Computing Studies, School of Computer Science, Peking University, Beijing 100871, China}
\date{\today}

\begin{abstract}
Quantum imaginary-time evolution (QITE) is a fundamental framework for preparing ground and thermal states, yet its computational cost scales significantly with the evolution duration $\tau$. Reducing this duration is critical for practical quantum advantage. Here, we establish a unified theoretical framework for the Mpemba effect in QITE—a counterintuitive phenomenon where {a state initially farther from the ground state relaxes to it faster than one initially closer}. We derive a remarkably simple necessary and sufficient condition for the occurrence of this effect, showing it is uniquely determined by the population {ratios of excited states} to the ground state. For practical state preparation, we introduce a rigorous sufficient condition for the finite-time Mpemba effect, ensuring the crossing occurs before reaching a prescribed proximity threshold. Furthermore, we unveil unique {dynamical} features, including a multiple-crossing phenomenon in multi-level systems and simultaneous intersections for collinear initial states. 
Our results provide criteria for identifying favorable initial states in QITE and offer deep insights into the speed limit of quantum state preparation.
\end{abstract}

\maketitle

The dynamics of quantum systems can be studied not only in real time, but also in imaginary time. Under the Wick rotation $t\to -i\tau$, an initial state $\ket{\Psi_0}$ evolves along the normalized trajectory $\ket{\Psi(\tau)} \propto e^{-H\tau}\ket{\Psi_0}$, progressively suppressing excited-state components in favor of the ground state. This process, known as quantum imaginary-time evolution (QITE), has become a powerful framework for ground-state preparation, open-system dynamics simulation, and finite-temperature calculations, with applications ranging from fundamental many-body physics to practical optimization problems~\cite{McArdle2019,Yeter2020,Gomes2020,Mao2020,Motta2020,Nishi2021, Sun2021, Kamakari2022, Gluza2025}.   

Unlike real-time quantum simulation, for which efficient and in some settings optimal algorithms are
known for broad classes of Hamiltonians~\cite{Low2017, Low2019}, no comparably general efficient algorithm is known for QITE. This difficulty reflects the QMA-hardness of ground-state preparation~\cite{Kempe2006}. Despite the diversity of existing QITE methods (e.g., variational approaches~\cite{McArdle2019, Yeter2020, Gomes2020}, unitary dilation and approximation schemes~\cite{Motta2020, Nishi2021, Sun2021}, weak-measurement-based protocols~\cite{Mao2020}, quantum signal processing techniques~\cite{Silva2023, Low2025, Suzuki2025QSP}, and recursive algorithms~\cite{Gluza2025, Son2025}), the cost of simulation generally grows rapidly with the target imaginary time. In particular, for recent general-purpose algorithms that deterministically approximate QITE for arbitrary duration~\cite{Gluza2025}, the cost scales exponentially with the target QITE duration, and this scaling is known to be optimal for certain Hamiltonian families~\cite{Suzuki2025}. Consequently, even a modest reduction in the required QITE duration can yield dramatic, and sometimes exponential, savings in computational cost. This raises a natural practical question: can one choose or engineer initial states that reach the target accuracy in a shorter imaginary time?

Naively, one might expect that the farther a state is from the ground state, the longer it should take to reach it under QITE, much as one might expect hotter water to take longer to freeze than colder water. The Mpemba effect, however, illustrates striking exceptions to this heuristic~\cite{Mpemba1969}. First popularized in macroscopic settings~\cite{Kell1969,Auerbach1998,Jeng2006,Katz2009,Vynnycky2015,Lu2017,Lasanta2017,Kumar2020}, this phenomenon has more recently attracted growing interest in open quantum systems~\cite{Carollo2021,Chatterjee2023,Chatterjee2024,Moroder2024,Wang2024,Nava2024,Shapira2024,David2025,Tan2025,Zhang2025,bao2025accelerating} {and real-time evolution in closed quantum system}~\cite{Rylands2024,Joshi2024,Murciano2024,Turkeshi2025,Ares2025,summer2026}. Meanwhile, analogous behavior has recently been observed \emph{numerically} for pure states undergoing QITE~\cite{Chang2024,Yu2025,Zander2026}. Leveraging such effects in practice, however, requires a theoretical understanding of why and when they occur, how they depend on the initial state, and when they remain operationally relevant.

Here, we develop a unified {theoretical} framework for Mpemba effects in QITE using a broad class of distance-like functions. Our results provide principled criteria for understanding when seemingly less favorable initial states can reach a target ground-state proximity sooner, thereby helping identify candidate initializations for faster QITE. We first derive a simple necessary-and-sufficient condition for the Mpemba effect in QITE, showing that it is completely determined by the population ratios of excited states to the ground state. We then consider the practically relevant setting of approximate state preparation, where evolution is halted once the state reaches a prescribed target accuracy, and establish a rigorous sufficient condition for observing a finite-time Mpemba effect. Finally, we show that initial states whose population vectors are collinear {all intersect simultaneously}, and identify a multiple-crossing phenomenon in multi-level systems through numerical simulations.

\vspace{0.1cm}
\noindent \emph{Framework.---} Suppose a Hamiltonian $H$ has $n$ energy levels with eigenstates $\{\ket{E_{0}}, \ket{E_{1}}, \dots, \ket{E_{n-1}}\}$ and corresponding eigenenergies $\{E_{0}, E_{1}, E_{2}, \dots, E_{n-1}\}$.
Assume that all eigenenergies are nondegenerate and strictly ordered ($E_{i}< E_{j}$ for $i< j$), with the ground-state energy set to $E_{0}=0$; see SI for justification. 
QITE of the initial state $\ket{\Psi_{0}}$ is then  given by  
\begin{equation}\label{eq: QITE def}
\left|\Psi(\tau)\right\rangle \coloneqq \frac{e^{-H\tau}\ket{\Psi_{0}}}{\sqrt{\bra{\Psi_{0}}e^{-2H\tau}\ket{\Psi_{0}}}},
\end{equation}
where the real number $\tau$ is the QITE duration.
QITE converges to the ground state by exponentially suppressing all the excited state components, whenever the initial state has a non-zero overlap with the ground state, i.e., $\braket*{\Psi_{0}}{E_{0}}\neq0$. 
An equivalent alternative formulation reveals that QITE in fact yields an optimal trajectory for this task.
It is known that QITE in Eq.~\eqref{eq: QITE def} solves the differential equation $\partial_{\tau} \Psi(\tau) = [[\Psi(\tau), H], \Psi(\tau)]$ when $\Psi(\tau) = \dyad{\Psi(\tau)}$ is a pure density matrix~\cite{Anshu2021}.
This differential equation, named double-bracket flow~\cite{Brockett1991}, has been identified as a Riemannian gradient flow~\cite{Hlemke1994}, where the cost function is the average energy $\bar{E}(\tau) \coloneqq \expval*{H}{\Psi(\tau)}$.
In other words, $\partial_{\tau}\Psi(\tau) = -\mathrm{grad}\bar{E}(\tau)$, and QITE can be regarded as a Riemannian optimization of the average energy reduction~\cite{Gluza2025, Suzuki2025, Mcmahon2025, Zander2026}.
This realization positions QITE as the best \emph{computational} cooling process for quantum states, albeit it is not physical.

Our goal is to compare the cooling performance of two different initial states undergoing QITE. 
To define cooling performance rigorously, one must choose a distance-like function that measures proximity to the {ground} state.
There are many different choices such as the {ground-state} infidelity and the average energy. 
To unify these choices, we define a general distance-like function based on a few minimal assumptions. First, we assume that it is given by an observable $\hat{f}(H) = \sum_{i=0}^{n-1} \fn(E_{i})\dyad{E_{i}}$, where $\fn(E)$ is a nondecreasing function of $E$. Second, we assume {$\fn(E_{1})>\fn(E_0)=0$} to make the distance zero if and only if the state reaches the ground state.
The general distance to the ground state is then given as
$\FN(\tau)\coloneqq \expval{\hat{f}(H)}{\Psi(\tau)} = \sum_{i=0}^{n-1}p_{i}(\tau)\fn(E_{i})$.
Here $p_{i}(\tau) \coloneqq \abs*{\braket*{E_{i}}{\Psi(\tau)}}^{2}$ are the energy‑level populations of the state $\ket{\Psi(\tau)}$.
Two canonical measures, namely infidelity and average energy, can be recovered by choosing appropriate functions $\fn$: 
infidelity by taking $\fn(0) = 0$ and $\fn(E) = 1$ for all $E>0$; average energy by taking $\fn(E) = E$.
Moreover, it is guaranteed that QITE of any positive duration will strictly decrease the distance, i.e. $\FN(\tau)>\FN(\tau')$ if $\tau'>\tau$.

For a general distance {function}, suppose that $\ket*{\psih_{0}}$ is hotter than $\ket*{\psic_{0}}$, that is $\FNh{0}>\FNc{0}$, the Mpemba effect happens for these states if there exists $\tau^{\star}>0$ such that $\FNh{\tau}<\FNc{\tau}$ for all $\tau>\tau^{\star}$.
Our definition mirrors the definition of the conventional Mpemba effects in physical cooling processes using distance functions between mixed states~\cite{Lu2017}.

\begin{figure*}
	\begin{centering}
		\includegraphics[scale=0.435]{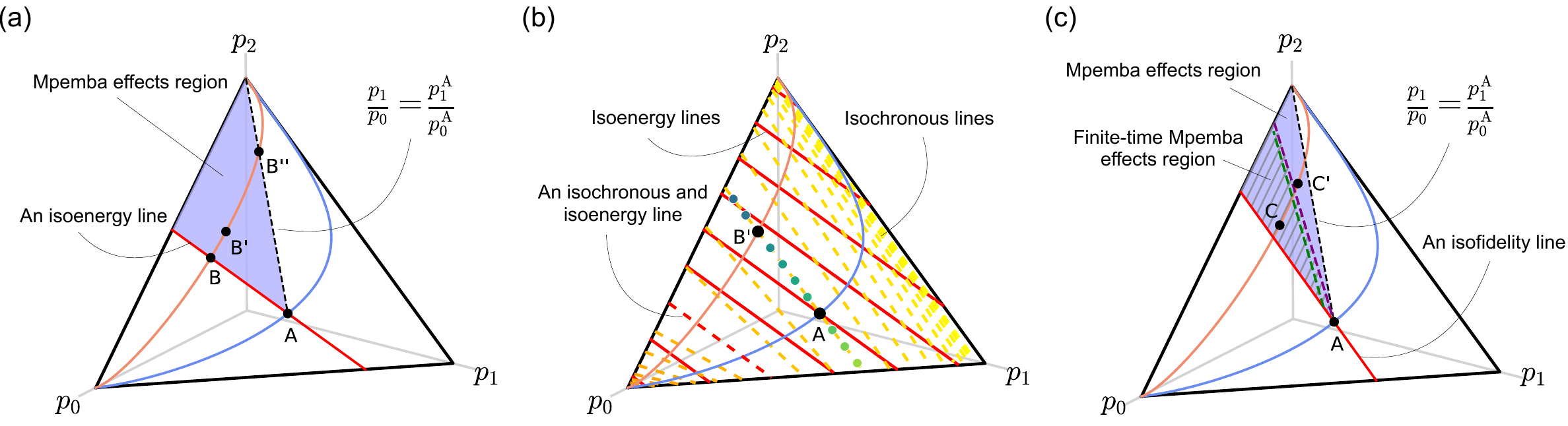}
	\par\end{centering}
		\caption{
		The population simplex of $(p_{0},p_{1},p_{2})$ for a three-level system, depicted as an equilateral triangle with black edges.
		The blue and peach curves are QITE trajectories. 
		Red lines are contours of constant distance to the ground state for a given Hamiltonian, measured by average energy in (a) and (b) and by infidelity in (c).
		The purple regions denote sets of initially hotter states that eventually become colder than the given initially colder state $A$ under QITE.
		\textbf{(a)} 
		Points $B$, $B'$, and $B''$ lie on the same QITE trajectory, with $B''$ located exactly on the boundary of the purple region.
		\textbf{(b)} 
		The set of initial states are marked by collinear green dots, where darker shades indicate higher average energy than lighter ones.
		The dashed lines represent isochronous lines, with colors shifting toward orange as $\tau$ increase and toward yellow as $\tau$ decrease. 
		The red dashed line highlights the coincidence of an isoenergy line with an isochronous line.
		\textbf{(c)} 
		The purple region with diagonal stripes is the set of all initially hotter states that cool faster than the given initially colder state $A$ with a threshold $\epsilon=0.05$. 
		The purple dashed line marks the numerically determined boundary for the finite-time Mpemba effect, while the green dashed line denotes the boundary derived from the sufficient condition Eq.~\eqref{eq: finite time Mpemba sufficient} of Theorem~\ref{thm: finite time sufficient}. 
		Points $C$ and $C'$ lie on the same QITE trajectory.
		}
	\label{Figure1}
\end{figure*}

\vspace{0.1cm}
\noindent \emph{Main results.---}
Our main {result provides a simple necessary and sufficient condition for the Mpemba effect with respect to any $\FN$ measure, thereby unifying all such measures under a single criterion. This can be viewed as the closed-system analogue of the unification of $f$-divergence Mpemba effects in Ref.~\cite{Tan2025}, with the key distinction that in our case the same necessary and sufficient condition applies to all \(\FN\) measures, unlike in Ref.~\cite{Tan2025}.}

\begin{theorem}\label{thm: Mpemba iff}
	Suppose that $\ket*{\psih_{0}}$ is ``hotter'' than $\ket*{\psic_{0}}$ with respect to the $\FN$, i.e. $\FNh{0} > \FNc{0}$.
	Let $i$ be the smallest index such that $\frac{\ph_{i}(0)}{\ph_{0}(0)} \neq \frac{\pc_{i}(0)}{\pc_{0}(0)}$.
	Then the Mpemba effect occurs if and only if $\frac{\ph_{i}(0)}{\ph_{0}(0)} < \frac{\pc_{i}(0)}{\pc_{0}(0)}$.
\end{theorem}

Theorem~\ref{thm: Mpemba iff} is proved in End Matter. {A protocol for redistributing the populations of an original initial state to obtain a state with higher $p_{1}/p_{0}$, serving as the initially hotter state, is present in {SI}}.
To visualise the Mpemba effect, we consider a geometric picture.
Any initial population vector $\vec{p}(0) = (p_{0}(0),p_{1}(0),\cdots,p_{n-1}(0))$ can be mapped onto an $(n-1)$-simplex.
Its QITE trajectory $\vec{p}(\tau)$ for $\tau\in[-\infty,\infty]$ driven by a chosen Hamiltonian is a line in this simplex whose endpoints at $\tau = \pm\infty$ are the vertices of the simplex corresponding to the lowest and the highest populated energy levels of the initial state $\vec{p}(0)$. 
Each trajectory thus represents a continuous one‑parameter family of states, with any two points on the same line connected by forward or backward QITE.
Moreover, distinct trajectories never intersect in the interior of the simplex, 
 so the simplex is partitioned into a cluster of disjoint lines.

Meanwhile, the $(n-1)$-simplex can be foliated into $(n-2)$-surfaces, each of which consists of states that share the same distance to the ground state.
These are, for example, isoenergy surfaces (when distance is measured by average energy) or isofidelity surfaces (when infidelity is used); see the red lines in Fig.~\ref{Figure1}.
Using Theorem~\ref{thm: Mpemba iff}, we can also visualise the set of initially hotter states that would eventually become colder than a fixed initially colder state after QITE.
Fig.~\ref{Figure1} (a) and (c) show these regions (indicated in purple) with respect to the initially colder state $A$. 
The Mpemba effect arises because each state approaches the ground state with different speed under QITE: higher-energy components are suppressed exponentially faster than lower-energy ones.
Consequently, even when two states start with the same distance to the ground state, after an QITE duration their distances differ.
This discrepancy is displayed in Fig.~\ref{Figure1} (b) as the mismatch between isoenergy lines and isochronous lines (viz.~the lines connecting $\{\vec{p}^{(l)}(\tau)\}_{l}$ starting from a set of initial states $\{\vec{p}^{(l)}(0)\}_{l}$ for equal $\tau$), both of which are plotted at equal intervals of energy and imaginary time, respectively. 
If all initial states lie on a straight line, then the shape of the isochronous lines of these points will not change under QITE. 
Instead, they rotate on the {population} simplex because of distinct evolution directions of each state and thus coincide with an isoenergy line at a certain imaginary time, {indicating that all the initial states attain the same average energy simultaneously (shown as} the red dashed line in Fig.~\ref{Figure1}(b)); see {SI} for a detailed proof.

\vspace{0.1cm}
\noindent \emph{Finite-time Mpemba effects.---}
Theorem~\ref{thm: Mpemba iff} demonstrates that the Mpemba effect in QITE is completely determined by the {population ratio of the lowest excited state with a different ratio to} the ground state.
This is because higher-level populations decay exponentially faster under QITE and become negligible after sufficiently long durations. 
However, this is an inherently long‑time statement about QITE.
To understand short‑ and intermediate‑time crossings one must analyse the finite‑time changes in the distance function $\FN$.

For practical purposes, it suffices to prepare a state that is close enough to the desired state.
In the context of QITE, this approximate nature is inevitable as the exact ground state preparation generally requires infinite QITE duration. 
Let $\epsilon$ be the threshold for the tolerable error for the ground state preparation. 
This means that QITE is terminated once the state $\ket{\Psi(\tau)}$ reaches $\FN(\tau) = \epsilon$.
Then, even when Theorem~\ref{thm: Mpemba iff} guarantees that the Mpemba effect will eventually occur if QITE continues indefinitely, QITE may halt before the Mpemba effect takes place.

To address this, we study {the sufficient condition for the finite-time} Mpemba effect. 
Given a hotter state $\ket*{\psih_{0}}$ and a colder state $\ket*{\psic_{0}}$ such that $\FNh{0}>\FNc{0}$, the Mpemba effect occurs before the threshold if there exists $\tau^{\star}$ such that 
$\FNh{\tau^{\star}}=\FNc{\tau^{\star}}>\epsilon$ and $\FNh{\tau}<\FNc{\tau}$ for all $\tau>\tau^{\star}$. 
{This sufficient condition depends on the choice of the distance function.}
In particular, we focus on the case where infidelity to the ground state $1-|\braket*{E_{0}}{\Psi(\tau)}|^{2} = \FI(\tau)$ is chosen as the distance {function}, as it is often the most important metric for state preparation tasks.
Besides, the sufficient condition for general $\FN$ measures {with $f(E_{n-1})> f(E_1)$} is presented in SI, which includes the case where average energy is chosen as the distance {function}.

\begin{theorem}\label{thm: finite time sufficient}
{Suppose the ground-state infidelity of $\ket*{\psih_{0}}$ is higher than that of $\ket*{\psic_{0}}$. Define} ratios $\rh \coloneqq \frac{\ph_{1}(0)}{\ph_{0}(0)}$, $\rc \coloneqq \frac{\pc_{1}(0)}{\pc_{0}(0)}$, and $\sh \coloneqq \frac{1-\ph_{0}(0)}{\ph_{0}(0)}$. The Mpemba effect occurs before the threshold $\epsilon$ if  $\rh<\rc$ and
	\begin{align}\label{eq: finite time Mpemba sufficient}
		\epsilon < \left(1 + \frac{1}{\rc}\left(\frac{\sh-\rh}{\rc-\rh}\right)^{\frac{E_{1}}{E_{2}-E_{1}}} \right)^{-1} .
	\end{align}
\end{theorem}

It provides an easy-to-check condition for the finite-time Mpemba effect. 
The {upper bound} of the threshold $\epsilon$ 
increases monotonically with $\rc$ while it decreases monotonically with $\rh$, $\sh$, and $\frac{E_{1}}{E_{2}-E_{1}}$.
For a sufficiently small prescribed threshold, the fastest way for a given initial state to relax to the ground state is to convert $p_{i}(0)$ to $0$ for $0<i<n-1$,
which yields a maximum acceleration time governed by $\Delta \tau(\epsilon)\approx -\frac{E_{n-1}-E_{k}}{2E_{k}E_{n-1}}\ln\epsilon+\frac{1}{2E_{k}}\ln\frac{p_{k}^{\text{c}}(0)f(E_{k})}{p_{0}^{\text{c}}(0)}+\frac{1}{2E_{n-1}}\ln\frac{p_{0}^{\text{h}}(0)}{p_{n-1}^{\text{h}}(0)f(E_{n-1})}$.
While we have required $\rh<\rc$ in Theorem~\ref{thm: finite time sufficient} for simplicity, the result can be generalized for the case $\rh = \rc$. 
We refer to End Matter for the proof and detailed discussion.

\begin{figure*}
\begin{centering}
\includegraphics[scale=0.4]{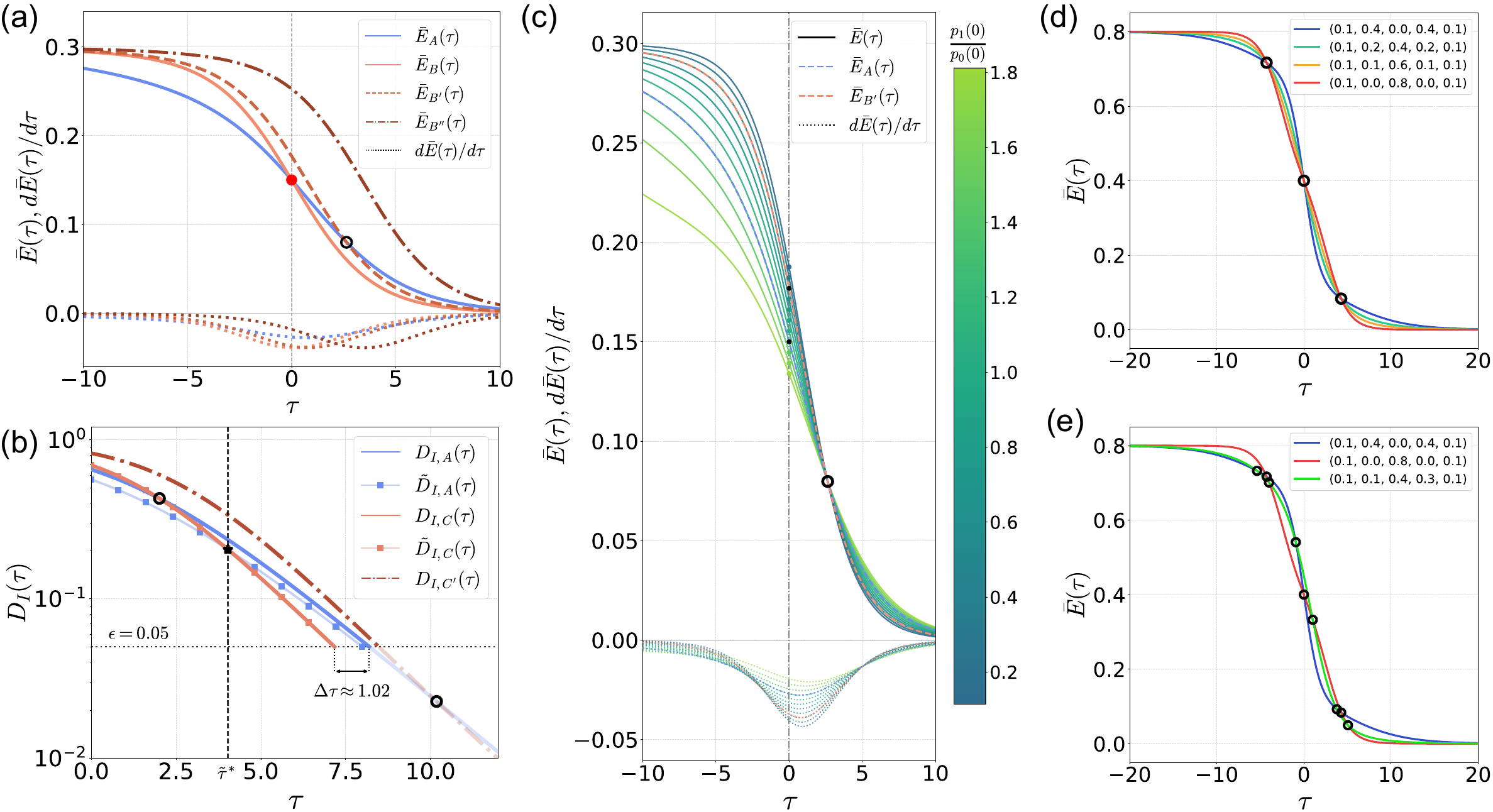}
\par\end{centering}
\caption{The images of evolving curves of distance ((a),(c)-(e): average energy; (b): infidelity) and their derivatives versus imaginary time. (a)-(c): $E_0=0, E_1=0.2, E_2=0.3$. \textbf{(a)} Initial states $A,\ B,\ B',\ B''$ correspond to the four points depicted in Fig.~\ref{Figure1}(a). The red point denotes point $A$ and point $B$ have equal initial average energy. The black circle denotes the curves $\bar{E}_{B'}(\tau)$ and $\bar{E}_{A}(\tau)$ cross at $\tau \approx2.64$ with $\bar{E}\approx0.08$.
\textbf{(b)} Initial states $A,\ C,\ C'$ correspond to the three points depicted in Fig.~\ref{Figure1}(c). The two black circles denote 
the curves $D_{I,C}(\tau)$ and $D_{I,A}(\tau)$ cross at $\tau \approx2.01$ with  $D_I\approx0.43$, and the curves $D_{I,C'}(\tau)$ and $D_{I,A}(\tau)$ cross at $\tau \approx10.18$ with $D_I\approx0.02$. The black star denotes the curves $\tilde{D}_{I,C}(\tau)$ and $\tilde{D}_{I,A}(\tau)$ cross at $\tilde\tau^* \approx4.03$ with $D_I\approx0.20$. 
\textbf{(c)} The color gradient from lighter to darker indicates decreasing values of the ratio $p_1(0)/ p_0(0)$.
The population distributions of states at $\tau=0$ are depicted in Fig.~\ref{Figure1}(b). The black circle shows all curves intersect with each other simultaneously.
(d),(e): The black circles denote crossings of evolving curves in a $5$-level system. The initial population vectors are marked in the legend. Here $E_0=0, E_1=0.15, E_2=0.4, E_3=0.65, E_4=0.8$. \textbf{(d)} All initial states are collinear in the population simplex. \textbf{(e)} Initial states are non-collinear in the population simplex.
\label{Figure2}}
\end{figure*}

A closed-form necessary and sufficient condition for the general case remains elusive. Even for a three-level system, such a condition would require solving equations of the form $Ax^{a}+Bx^{b} = C$.
Nonetheless, the sufficient condition presented in Theorem~\ref{thm: finite time sufficient} is highly robust for two primary reasons. First, it is universal: whenever the Mpemba effect is predicted in the long-time limit (i.e., $\rc > \rh$), the theorem guarantees a non-empty range for the threshold $\epsilon$ within which the effect must occur. Thus, while it may not capture the entire permissible range of $\epsilon$, it ensures the phenomenon is never missed entirely. Second, the condition is tight in the worst-case scenario: it becomes both necessary and sufficient 
when $\ph_{i}(0) = 0$ for all $i>2$ and $\pc_{j}(0) = 0$ for all $j>1$.

Fig.~\ref{Figure1}(c) compares the Mpemba effect in long-time and finite-time regimes. The purple region identifies all initially hotter states that eventually become colder than the given initially colder state A as $\tau \to \infty$. When a practical infidelity threshold of $\epsilon=0.05$ is introduced, a regime where QITE is expected to outperform alternatives like quantum phase estimation~\cite{Gluza2025}, the set of valid states shrinks to the diagonally-striped area.
Furthermore, the sufficient condition in Theorem~\ref{thm: finite time sufficient} proves remarkably tight: the predicted boundary (green dashed line) {is in close proximity to} the numerical finite-time boundary (purple dashed line), even beyond the theoretical worst-case scenario.

\vspace{0.1cm}
\noindent \emph{Numerical results.---}
Now we show numerical results to {demonstrate above} theory. 
In Fig.~\ref{Figure2}(a), the four evolution curves are derived from the four initial states located at the four points $A=(0.35, 0.45, 0.20),\ B=(0.45, 0.15, 0.40),\ B'=(0.36, 0.16, 0.48),$ and $ B''=(0.11, 0.14, 0.75)$  depicted in Fig.~\ref{Figure1}(a). 
Since point $B'$ is situated within the purple region in Fig.~\ref{Figure1}(a), curves $\bar{E}_{B'}(\tau)$ and $\bar{E}_{A}(\tau)$ are guaranteed to cross at a later time, triggering the Mpemba effect. Notably, the curves $\bar{E}_{B'}(\tau)$ and $\bar{E}_{B''}(\tau)$ are essentially right-shifted versions of the curve $\bar{E}_{B}(\tau)$. Since point $B''$ is positioned on the boundary of the purple region, curve $\bar{E}_{B''}(\tau)$ represents the limit of this rightward shift while still ensuring the occurrence of the crossing with curve $\bar{E}_{A}(\tau)$. {Furthermore, the lower panel of Fig.~\ref{Figure2}(a) displays the corresponding time-derivative curves of the average energy, revealing that the decay rate of the average energy is not necessarily correlated with its magnitude.}

Fig.~\ref{Figure2}(b) shows the evolution curves for the three initial states located at the three points $A=(0.35, 0.45, 0.2),\ C=(0.31, 0.16, 0.53),$ and $ C'=(0.18, 0.15, 0.67)$ depicted in Fig.~\ref{Figure1}(c). Since point $C$ lies within the diagonally-striped region, curve $D_{I,C}(\tau)$ necessarily intersects with curve $D_{I,A}(\tau)$ within a sufficiently short time, {successfully} manifesting the finite-time Mpemba effect. {Notably, curve $D_{I,C}(\tau)$ reaches the threshold faster than curve $D_{I,A}(\tau)$ with a time difference $\Delta\tau \approx 1.02$, providing a clear quantitative demonstration of this acceleration phenomenon. In contrast, while} point $C'$ sits within the purple region, it lies outside the diagonally-striped area. {Therefore, its evolution is truncated upon reaching the threshold $\epsilon$ prior to an intersection, rendering the Mpemba effect unobservable in practice, as illustrated by the semi-transparent curves extending below the threshold in Fig.~\ref{Figure2}(b). Furthermore, Fig.~\ref{Figure2}(b)} shows curves $\tilde{D}_{I,A}$ and $\tilde{D}_{I,C}$ defined in the proof of Theorem~\ref{thm: finite time sufficient}. {The observation that their intersection at $\tilde{\tau}^*$ occurs strictly above the threshold firmly validates our proposed sufficient condition for finite-time Mpemba effects.}

To exhibit the change of population distributions of states under QITE, we plot the 
isochronous lines in Fig.~\ref{Figure1}(b). Choosing points situated on a straight line contains both points $A$ and $B'$ as initial states, we plot the corresponding evolution curves of average energies in Fig.~\ref{Figure2}(c). It is evident that all curves intersect with each other at $\tau \approx2.64$ with average energy $\bar{E}\approx0.08$, which is consistent with the red dotted line in Fig.~\ref{Figure1}(b).

For multi($n>3$)-level systems, the imaginary time evolution curves become more diverse. Unlike the 3-level case shown in Fig.~\ref{Figure2}(c), where evolution curves only intersect once, in multi-level systems, such curves may undergo multiple intersections. 
Taking the 5-level system as an example, we plot 
the images of 
curves $\bar{E}(\tau)$ with different initial population vectors in Fig.~\ref{Figure2}(d) and (e). It shows that these curves intersect one another more than once. This indicates that the mere existence of crossings is insufficient to conclude the occurrence of the Mpamba effect, as exemplified by the green and red lines in Fig.~\ref{Figure2}(e). Interestingly, for a fixed initial average energy and $n-3$ {arbitrarily chosen but fixed} initial populations $p_i$, all corresponding QITE curves $\bar{E}(\tau)$ appear to always intersect simultaneously, as shown in Fig.~\ref{Figure2}(d).  This is because the population vectors of these initial states are colinear and preserve their collinearity throughout the evolution. Consequently, all population vectors reach the same isoenergy surface at the same time, causing all curves to intersect simultaneously. If the initial population vectors are not colinear, evolution curves will intersect with each other at different $\tau$, as shown in Fig.~\ref{Figure2}(e).

\vspace{0.1cm}
\noindent \emph{Discussion.---}
In summary, we have established the condition for the occurrence of the Mpemba effects in quantum imaginary-time evolution with infinite and finite times. 
Our result singles out the ratio between the first excited and ground state populations as the main deciding factor of QITE performance.
This implies that the QITE of a generic given state in long-time regime resembles that of an effectively qubit state having only the ground and excited state populations with the same ratio as the original state. 
This simplification would alleviate the difficulty of analysing QITE processes for general initial states and Hamiltonians.
In particular, it has been found that the Grover search problem~\cite{Grover1996} can be regarded as an approximate QITE for such effectively qubit states~\cite{Suzuki2025}.
Combined with our finding, it may be possible to {characterize} the algorithmic complexity of QITE to the well-understood complexity analysis of Grover search problems~\cite{Bennett1997, Zalka1999} {at least for the long-time regime}. 
Ultimately, our results highlight a principled method that leverages the Mpemba effect to optimize initial states, thereby accelerating ground-state preparation. With the rapid progress of quantum computation, the realization of Mpemba effects in imaginary time evolution on quantum hardware holds broad prospects as well as presents exciting opportunities for both fundamental and applied research.

\vspace{0.1cm}
\noindent \emph{{Acknowledgments.}---}
This work is supported by 
{Beijing Natural Science Foundation (Grant No.~1264065 and Z250004),}
the National Natural Science Foundation of China Grant (No.~12361161602), 
NSAF (Grant No.~U2330201), 
Quantum Science and Technology-National Science and Technology Major Project (2023ZD0300200), Beijing Science and Technology Planning Project (Grant No.~Z25110100810000), the High-performance Computing Platform of Peking University, the National Research Foundation through the NRF Investigatorship on Quantum-Enhanced Agents (Grant No. NRF-NRFI09-0010) and the National Quantum Office, hosted in A*STAR, under its Centre for Quantum Technologies Funding Initiative (S24Q2d0009), the Singapore Ministry of Education Tier 1 Grant RT4/23 and RG91/25. JS is supported by the start-up grant of the Nanyang Assistant Professorship at the Nanyang Technological University in Singapore awarded to Nelly Ng.

\section*{End Matter}

\subsection*{Proof of Theorem~\ref{thm: Mpemba iff}}
%%%%%%%%%%%%%%%%%%%%%%%%%%%%%%%%%%%%%%%%%%%%%%%%%%%%%%%%%%%%%%%%%%
\begin{proof}
	Let us denote $\rh = \frac{\ph_{1}(0)}{\ph_{0}(0)}$ and $\rc = \frac{\pc_{1}(0)}{\pc_{0}(0)}$. 
	Furthermore, assume that $\rh\neq\rc$, as it is straightforward to extend our proof for the case $\rh = \rc$.
	We defer the case of $\rh =0$ or $\rc = 0$ to SI.
	Hence the necessary and sufficient condition we want to prove is $\rh<\rc$.
	
	The population vectors after the QITE duration $\tau$ can be calculated directly from Eq.~\eqref{eq: QITE def} as
	\begin{align}
		\ph_{i}(\tau) &= \frac{\ph_{i}(0)e^{-2E_{i}\tau}}{\sum_{j=0}^{n-1}\ph_{j}(0)e^{-2E_{j}\tau}},\label{eq: hot pop i}\\ 
		\pc_{i}(\tau) &= \frac{\pc_{i}(0)e^{-2E_{i}\tau}}{\sum_{j=0}^{n-1}\pc_{j}(0)e^{-2E_{j}\tau}}.\label{eq: cold pop i}
	\end{align}

	\textbf{Proof of necessity:}
	Suppose that the condition in Theorem~\ref{thm: Mpemba iff} is not met, i.e. $\rh > \rc$, while by assumption $\FNh{0}>\FNc{0}$.
	Using Eqs.~\eqref{eq: hot pop i} and~\eqref{eq: cold pop i}, we write
	\begin{align}\label{eq: energ nec ineq 1}
		\FNx{\tau} &= \frac{\sum_{i}p^{\mathrm{x}}_{i}(0)\fn(E_{i})e^{-2E_{i}\tau}}{\sum_{j}p^{\mathrm{x}}_{j}(0)e^{-2E_{i}\tau}}\nonumber\\
		&= \frac{r^{\mathrm{x}}\fn(E_{1})e^{-2E_{1}\tau}+\eta^{\mathrm{x}}(\tau)e^{-2E_{2}\tau}}{1+r^{\mathrm{x}}e^{-2E_{1}\tau} + o^{\mathrm{x}}(\tau)e^{-2E_{2}\tau}},
	\end{align}
	for $\mathrm{x} = \mathrm{h}, \mathrm{c}$, where we define two functions 
	\begin{align}
		o^{\mathrm{x}}(\tau) &\coloneqq  \sum_{i=2}^{n-1}\frac{p^{\mathrm{x}}_{i}(0)}{p^{\mathrm{x}}_{0}(0)}e^{-2(E_{i}-E_{2})\tau},\label{eq: o function def}\\
		\eta^{\mathrm{x}}(\tau) &\coloneqq \sum_{i=2}^{n-1}\frac{p^{\mathrm{x}}_{i}(0)\fn(E_{i})e^{-2(E_{i}-E_{2})\tau}}{p^{\mathrm{x}}_{0}(0)}.\label{eq: eta function def}
	\end{align} 
	Since $E_{i}-E_{2}\geq0$ for all $i$ in the summand of Eqs.~\eqref{eq: o function def} and~\eqref{eq: eta function def}, we find that both $o^{\mathrm{x}}(\tau)$ and $\eta^{\mathrm{x}(\tau)}$ are bounded in finite ranges between two nonnegative numbers for all $\tau\geq0$, and they decrease monotonically with $\tau$.
	Now we rewrite Eq.~\eqref{eq: energ nec ineq 1} as
	\begin{align}
		\FNx{\tau} = R^{\mathrm{x}}(\tau)e^{-2E_{1}\tau} + \mu^{\mathrm{x}}(\tau)e^{-2E_{2}\tau},
	\end{align}
	where
	\begin{align}
		R^{\mathrm{x}}(\tau) &\coloneqq \frac{r^{\mathrm{x}}\fn(E_{1})}{1+r^{\mathrm{x}}e^{-2E_{1}\tau}+ o^{\mathrm{x}}(\tau)e^{-2E_{2}\tau}},\\
		\mu^{\mathrm{x}}(\tau) &\coloneqq \frac{\eta^{\mathrm{x}}(\tau) }{1+r^{\mathrm{x}}e^{-2E_{1}\tau} + o^{\mathrm{x}}(\tau)e^{-2E_{2}\tau}}.
	\end{align}
	Note that $R^{\mathrm{x}}(\tau)$ and $\mu^{\mathrm{x}}(\tau)$ are also bounded in a finite range, i.e. $R^{\mathrm{x}}(\tau)\in[R^{\mathrm{x}}_{\min},R^{\mathrm{x}}_{\max}]$ and $\mu^{\mathrm{x}}(\tau)\in[\mu^{\mathrm{x}}_{\min},\mu^{\mathrm{x}}_{\max}]$ for all $\tau\geq0$, where $R^{\mathrm{x}}_{\min},R^{\mathrm{x}}_{\max},\mu^{\mathrm{x}}_{\min},\mu^{\mathrm{x}}_{\max}$ are finite numbers independent of $\tau$.
	Then, for any $\varepsilon_{1}>0$, there exists $\tau_{1}$ such that 
	\begin{align}\label{eq: energ nec ineq 2}
		R^{\mathrm{x}}(\tau)(1- \varepsilon_{1}) < \FNx{\tau} e^{2E_{1}\tau} < R^{\mathrm{x}}(\tau) (1+ \varepsilon_{1}),
	\end{align}
	for both $\mathrm{x} = \mathrm{h},\mathrm{c}$ for all $\tau>\tau_{1}$.
	Eq.~\eqref{eq: energ nec ineq 2} implies that 
	\begin{align}\label{eq: energ nec ineq 3}
		\frac{R^{\mathrm{h}}(\tau)}{R^{\mathrm{c}}(\tau)}\frac{1-\varepsilon_{1}}{1+\varepsilon_{1}}<\frac{\FNh{\tau}}{\FNc{\tau}}<\frac{R^{\mathrm{h}}(\tau)}{R^{\mathrm{c}}(\tau)}\frac{1+\varepsilon_{1}}{1-\varepsilon_{1}}.
	\end{align}
	Similarly, for any $\varepsilon_{2}>0$, there exists $\tau_{2}$ such that
	\begin{align}
		\frac{\rh}{\rc}\frac{1-\varepsilon_{2}}{1+\varepsilon_{2}}<\frac{R^{\mathrm{h}}(\tau)}{R^{\mathrm{c}}(\tau)}<\frac{\rh}{\rc}\frac{1+\varepsilon_{2}}{1-\varepsilon_{2}}
	\end{align}
	for all $\tau > \tau_{2}$.
	Combined, we arrive at the statement that for any $\varepsilon_{3}>0$, there exists $\tau_{3}$ such that
	\begin{align}\label{eq: energ nec ineq 4}
		\frac{\rh}{\rc}\frac{1-\varepsilon_{3}}{1+\varepsilon_{3}}<\frac{\FNh{\tau}}{\FNc{\tau}}<\frac{\rh}{\rc}\frac{1+\varepsilon_{3}}{1-\varepsilon_{3}}
	\end{align}
	Since $\frac{\rh}{\rc}>1$ by assumption, by choosing sufficiently small $\varepsilon_{3}$, we can prove that $\FNh{\tau}>\FNc{\tau}$ for all $\tau>\tau_{3}$.
	
	\textbf{Proof of sufficiency:}
	Now suppose that $\rh<\rc$.
	Using Eq.~\eqref{eq: energ nec ineq 4} and choosing sufficiently small $\varepsilon_{3}$, we prove that there exists $\tau_{3}$ such that $\FNh{\tau}<\FNc{\tau}$ for all $\tau>\tau_{3}$.
\end{proof}
%%%%%%%%%%%%%%%%%%%%%%%%%%%%%%%%%%%%%%%%%%%%%%%%%%%%%%%%%%%%%%%%%%

\subsection*{Proof of Theorem~\ref{thm: finite time sufficient}}

%%%%%%%%%%%%%%%%%%%%%%%%%%%%%%%%%%%%%%%%%%%%%%%%%%%%%%%%%%%%%%%%%%
\begin{proof}
	The RHS of Eq.~\eqref{eq: finite time Mpemba sufficient} depends only on the first two populations from each initial state: $\ph_{0}(0),\ph_{1}(0)$ and $\pc_{0}(0),\pc_{1}(0)$.
	Since the sufficient condition must hold for any two states with those two first populations, we consider the worst-case scenario. 
	Define $\ket*{\psith_{0}}$ as a state with populations
	\begin{align}
		\pth_{0}(0) &= \ph_{0}(0),\quad \pth_{1}(0) = \ph_{1}(0),\\ 
		\pth_{2}(0) &= 1-\ph_{0}(0)-\ph_{1}(0),
	\end{align}
	$\pth_{i} (0)= 0$ for all other levels $i$.
	Then it is straightforward to check that 
	\begin{align}\label{eq: ph lower bound}
		\ph_{0}(\tau) = \frac{\ph_{0}(0)}{\sum_{i}\ph_{i}(0)e^{-2E_{i}\tau}} \geq \pth_{0}(\tau) =  \frac{\pth_{0}(0)}{\sum_{i}\pth_{i}(0)e^{-2E_{i}\tau}}.
	\end{align}
	Hence, the infidelity to the ground state of our actual state $\ket*{\psih(\tau)}$ is never larger than that of $\ket*{\psith(\tau)}$, i.e. $\FIh{\tau}\leq \tilde{D}_{I}^\text{h}(\tau) \coloneqq 1-\pth_{0}(\tau)$ for all $\tau\geq0$.
	The population $\pth_{0}(\tau)$ can also be written explicitly as 
	\begin{align}
		\pth_{0}(\tau) = \frac{\ph_{0}(0)}{\ph_{0}(0)+\ph_{1}(0)e^{-2E_{1}\tau}+(1-\ph_{0}(0)-\ph_{1}(0))e^{-2E_{2}\tau}}.
	\end{align}
	
	Similarly, consider the colder state evolution and define the quantity $\ptc_{0}(\tau)$ as 
	\begin{align}\label{eq: pc upper bound}
		\pc_{0}(\tau) =  \frac{\pc_{0}(0)}{\sum_{i}\pc_{i}(0)e^{-2E_{i}\tau}} \leq \frac{\pc_{0}(0)}{\pc_{0} (0)+ \pc_{1}(0)e^{-2E_{1}\tau}} \eqqcolon \ptc_{0}(\tau).
	\end{align}
	Now we have a lower bound for the colder state infidelity $\FIc{\tau}\geq\tilde{D}_{I}^\text{c}(\tau) \coloneqq 1-\ptc_{0}(\tau)$ for all $\tau\geq0$.
	
	The main strategy of the proof is to compare $\pth_{0}(\tau)$ and $\ptc_{0}(\tau)$. 
	It is more conducive to work with the expressions
	\begin{align}
		\frac{1}{\pth_{0}(\tau)} &= 1+\rh e^{-2E_{1}\tau}+(\sh-\rh) e^{-2E_{2}\tau},\label{eq: pth in rh sh}\\
		\frac{1}{\ptc_{0}(\tau)} &= 1+\rc e^{-2E_{1}\tau},\label{eq: ptc in rc}
	\end{align}
	where $\sh = \frac{1-\ph_{0}(0)}{\ph_{0}(0)}$ is the ratio between the all excited states populations combined and
    the ground state population. 
	Initially, hotter state has a smaller ground state population, which implies that 
	\begin{align}\label{eq: p0 inverses ordering}
		\frac{1}{\pth_{0}(0)} \geq \frac{1}{\ph_{0}(0)} > \frac{1}{\pc_{0}(0)} \geq \frac{1}{\ptc_{0}(0)}.
	\end{align}
	Define the difference $\Delta(\tau) \coloneqq \frac{1}{\pth_{0}(\tau)} - \frac{1}{\ptc_{0}(\tau)}$, with $\Delta(0)>0$ from Eq.~\eqref{eq: p0 inverses ordering}.
	We also assume $\rh<\rc$ for the Mpemba effect to happen at some point and $\Delta(\tau)<0$ for sufficiently large $\tau$ by Theorem~\ref{thm: Mpemba iff}.
	Then we show that there exists a unique solution $\tau = \tilde{\tau}^{\star}$ such that $\pth_{0}(\tau)$ and $\ptc_{0}(\tau)$ crosses, i.e.  
	\begin{align}
		\Delta(\tau) = (\rh - \rc)e^{-2E_{1}\tau} + (\sh-\rh)e^{-2E_{2}\tau} = 0,
	\end{align}
	only when
	\begin{align}\label{eq: tilde tau star def}
		\tau = \frac{1}{2(E_{2}-E_{1})}\log\left(\frac{\sh-\rh}{\rc-\rh}\right) \eqqcolon \tilde{\tau}^{\star}.
\end{align}
	Therefore, for any $\tau>\tilde{\tau}^{\star}$, it is guaranteed that $\Delta(\tau)<0$ and thus
	\begin{align}
		\ph_{0}(\tau) \geq \pth_{0}(\tau) > \ptc_{0}(\tau) \geq \pc_{0}(\tau).
	\end{align}
	In other words, there exists $\tau^{\star}\leq \tilde{\tau}^{\star}$ such that for all $\tau>\tau^{\star}$, $\FIh{\tau}<\FIc{\tau}$ and $\FI^\text{h}(\tau^{\star}) = \FI^\text{c}(\tau^{\star})$. 

	Now to prove that the Mpemba effect occurs before the finite threshold $\epsilon$, we only need to prove that at $\FI^\text{h}(\tau^{\star}) = \FI^\text{c}(\tau^{\star})>\epsilon$.
	Recall that 
	\begin{align}
		\FI^\text{c}(\tau^{\star}) \geq \FI^\text{c}(\tilde{\tau}^{\star}) \geq 1-\ptc_{0}(\tilde{\tau}^{\star}) = 1-\pth_{0}(\tilde{\tau}^{\star}).
	\end{align}
	The first inequality follows from the monotonicity of the infidelity and $\tau^{\star}\leq \tilde{\tau}^{\star}$; the second inequality follows from Eq.~\eqref{eq: pc upper bound}; and the last equality follows from $\Delta(\tau^{\star}) = 0$.
	Thus for the claim $\FI^\text{c}(\tau^{\star})>\epsilon$, it suffices to prove $1-\pth_{0}(\tilde{\tau}^{\star}) = 1-\ptc_{0}(\tilde{\tau}^{\star}) > \epsilon$.
	By substituting Eq.~\eqref{eq: tilde tau star def} into Eq.~\eqref{eq: pc upper bound}, we find that $1-\ptc_{0}(\tilde{\tau}^{\star}) > \epsilon$ is equivalent to Eq.~\eqref{eq: finite time Mpemba sufficient}.

    For the case $\rh = \rc$, suppose that $\frac{\ph_{j}(0)}{\ph_{0}(0)} < \frac{\pc_{j}(0)}{\pc_{0}(0)}$ where $j$ is the smallest index such that $\frac{\ph_{j}(0)}{\ph_{0}(0)} \neq \frac{\pc_{j}(0)}{\pc_{0}(0)}$. By simply replacing $r^{\mathrm{x}}$ with $\sum_{i=1}^{j} \frac{p^\mathrm{x}_{i}(0)}{p^\mathrm{x}_{0}(0)}$ and $E_{1},E_{2}$ with $E_{j},E_{j+1}$ in Eq.~\eqref{eq: tilde tau star def}, we obtain a new crossing time upper bound $\tilde{\tau}^{\star}$ and corresponding $\epsilon<1-\ptc_{0}(\tilde{\tau}^{\star})  $.
\end{proof}
%%%%%%%%%%%%%%%%%%%%%%%%%%%%%%%%%%%%%%%%%%%%%%%%%%%%%%%%%%%%%%%%%%

\subsection*{The derivative of maximum acceleration time}

%%%%%%%%%%%%%%%%%%%%%%%%%%%%%%%%%%%%%%%%%%%%%%%%%%%%%%%%%%%%%%%%%%
Based on the previously established necessary and sufficient conditions for the Mpemba effect, consider an initial state $\ket*{\psi^{\text{c}}_{0}}$ with population distribution $(p^{\text{c}}_0(0),\, p^{\text{c}}_1(0),\, p^{\text{c}}_2(0),\, \dots,\, p^{\text{c}}_{n-1}(0))$. Let $\ket*{\psi^{\text{h},i}_{0}}$ denote an initial state whose population distribution satisfies $p^{\text{h},i}_{j}(0) = 0$ for $0 < j < i$. It can be readily shown that there exists $\tau^*>0$ such that:
\begin{equation}
    \ket*{\psi^{\text{c}}(\tau)} > \ket*{\psi^{\text{h},k}(\tau)} > \ket*{\psi^{\text{h},k+1}(\tau)} > \dots > \ket*{\psi^{\text{h},n-1}(\tau)}
\end{equation}
 for all $\tau > \tau^*$, where $k$ is the index of the lowest nonzero-occupied excited state of $\ket*{\psi^{\text{c}}_{0}}$. Therefore, upon choosing $\ket*{\psi^{\text{h},n-1}_{0}}$ as the hotter initial state satisfying $\FNh{0} > \FNc{0}$, one can determine the relaxation times $\tau^{\text{h}}(\epsilon)$ and $\tau^{\text{c}}(\epsilon)$ for the two evolving states to reach a sufficiently small prescribed threshold $\epsilon$ which satisfies Eq.~(\ref{eq: finite time Mpemba sufficient}). From $D_f^{\text{h}}(\tau^{\text{h}}(\epsilon)) = \epsilon$, it follows that 
 \begin{align}
     &p^{\text{h}}_{n-1}(\tau^{\text{h}}(\epsilon))\fn(E_{n-1}) =\epsilon\\
\Rightarrow & 
\frac{\ph_{n-1}(0)e^{-2E_{n-1}\tau^{\text{h}}(\epsilon)}}{\ph_{0}(0)}\fn(E_{n-1})\approx \epsilon\\
\Rightarrow & 
\tau^{\text{h}}(\epsilon)\approx- \frac{1}{2E_{n-1}}\ln{\frac{\ph_{0}(0)\epsilon}{\ph_{n-1}(0)\fn(E_{n-1})}}.
 \end{align}
 From $D_f^{\text{c}}(\tau^{\text{c}}(\epsilon)) = \epsilon$, it follows that 
 \begin{align}
     &\sum_{i=0}^{n-1} p^{\text{c}}_{i}(\tau^{\text{c}}(\epsilon))\fn(E_{i})=\epsilon\\
\Rightarrow &  p^{\text{c}}_{k}(\tau^{\text{c}}(\epsilon))\fn(E_{k})\approx \epsilon \\
\Rightarrow & 
 \frac{\pc_{k}(0)e^{-2E_{k}\tau^{\text{c}}(\epsilon)}}{\pc_{0}(0)}\fn(E_{k})\approx\epsilon\\
\Rightarrow & 
\tau^{\text{c}}(\epsilon)\approx- \frac{1}{2E_{k}}\ln{\frac{\pc_{0}(0)\epsilon}{\pc_{k}(0)\fn(E_{k})}}.
 \end{align}
As a result, the fastest way for a given initial state to relax to the ground state while exhibiting the Mpemba effect is to set $p^{\text{h}}_1(0)=p^{\text{h}}_2(0)= \dots= p^{\text{h}}_{n-2}(0)=0$ and slightly increase the initial distance $D^{\text{h}}_f(0)$. This configuration yields a maximum acceleration time governed by $\Delta \tau(\epsilon)=\tau^{\text{c}}(\epsilon)-\tau^{\text{h}}(\epsilon)\approx -\frac{E_{n-1}-E_{k}}{2E_{k}E_{n-1}}\ln\epsilon+\frac{1}{2E_{k}}\ln\frac{p_{k}^{\text{c}}(0)f(E_{k})}{p_{0}^{\text{c}}(0)}+\frac{1}{2E_{n-1}}\ln\frac{p_{0}^{\text{h}}(0)}{p_{n-1}^{\text{h}}(0)f(E_{n-1})}$.
%%%%%%%%%%%%%%%%%%%%%%%%%%%%%%%%%%%%%%%%%%%%%%%%%%%%%%%%%%%%%%%%%%

\bibliography{Mpemba.bib}

\end{document}

\title{Supplemetary Information: Theory of Quantum Imaginary-Time Mpemba Effect}
\author{Yumeng Zeng}
\thanks{These authors contributed equally to this work.}
\affiliation{Center on Frontiers of Computing Studies, School of Computer Science, Peking University, Beijing 100871, China}

\author{Jeongrak Son}
\thanks{These authors contributed equally to this work.}
\affiliation{Nanyang Quantum Hub, School of Physical and Mathematical Sciences, Nanyang Technological University, 637371, Singapore}

\author{Mile Gu}
\email{mgu@quantumcomplexity.org}
\affiliation{Nanyang Quantum Hub, School of Physical and Mathematical Sciences, Nanyang Technological University, 637371, Singapore}
\affiliation{Center for Quantum Technologies, Nanyang Technological University, 639798, Singapore}
\affiliation{MajuLab, CNRS-UNS-NUS-NTU International Joint Research Unit, UMI 3654, Singapore 117543, Singapore}

\author{Xiao Yuan}
\email{xiaoyuan@pku.edu.cn}
\affiliation{Center on Frontiers of Computing Studies, School of Computer Science, Peking University, Beijing 100871, China}
\date{\today}

\maketitle

\section{Degenerate energy levels}\label{app: degeneracy}
In this section, we show that it is safe to assume that all eigenenergies $\{E_{0},E_{1},E_{2},\cdots,E_{n-1}\}$ are nondegenerate. 
Suppose that there exists a degeneracy, i.e. $E_{i} = E_{j}$ for some $i\neq j$.
Any pure state $\ket{\Psi_{0}}$ can be written as $\ket{\Psi_{0}} = c_{i}(0)\ket{E_{i}} + c_{j}(0)\ket{E_{j}} + c^{\perp}(0)\ket*{\psi^{\perp}}$ where $\ket*{\psi^{\perp}}$ is orthogonal to both $\ket{E_{i}}$ and $\ket{E_{j}}$.
Now consider another state 
\begin{align}
	\ket{\tilde{\Psi}_{0}} = \sqrt{|c_{i}(0)|^{2}+|c_{j}(0)|^{2}}\ket{E_{i}} + c^{\perp}(0)\ket{\psi^{\perp}}.
\end{align}
We prove that for all $\tau$, and for any observable that is a function only of Hamiltonian $\hat{f}(H)$, two states $\ket{\Psi(\tau)}$ and $\ket*{\tilde{\Psi}(\tau)}$ yields the same outcome. 
This comes from the simple observation that 
\begin{align}
	\langle E_{i} |\hat{f}(H)| E_{i}\rangle = \langle E_{j} |\hat{f}(H)| E_{j}\rangle,\quad \langle E_{i} |\hat{f}(H)| E_{j}\rangle = 0.
\end{align}
Hence, any degeneracy in the energy level can be effectively removed for any purpose relevant to this work.

\section{A protocol for redistributing the populations of an original initial state}\label{protocol}

To give a protocol for the setup of different initial states,
we consider the one-dimensional XYZ spin chain with periodic boundary conditions:
\begin{equation}
    H = -\sum_{j=1}^{L} \left( \sigma_j^x \sigma_{j+1}^x + \gamma\, \sigma_j^y \sigma_{j+1}^y + \mu\, \sigma_j^z \sigma_{j+1}^z \right),
\end{equation}
where $L$ is the number of lattice sites, $\sigma_j^{\alpha}$ ($\alpha = x, y, z$) are the Pauli operators acting on site $j$. Here the periodic boundary condition $\sigma_{L+1}^{\alpha} \equiv \sigma_1^{\alpha}$ is applied. For $\gamma=1$, the Hamiltonian have $U(1)$-symmetry. We set the tilted ferromagnetic state as the original initial state, which is constructed by applying a uniform $y$-rotation to the fully polarized state $\ket{000\cdots 0}$:
\begin{equation}
    \ket{\psi(\theta)} = e^{-i \frac{\theta}{2} \sum_j \sigma_j^y} \ket{000\cdots 0}.
\end{equation}
For the original initial state (or the initially colder state), $\ket{\Psi^{\mathrm{c}}(0)}=\ket{\psi(\theta)}$ is directly evolved under the Hamiltonian $H_1$ with $\gamma_1=1$. For the initially hotter state, we implement a forward-backward pre-evolution to obtain a state with a different population distribution from the original initial state. That is, the original initial state $\ket{\psi(\theta)}$ first undergoes a forward imaginary time evolution under the Hamiltonian $H_0$ with $\gamma_0\neq1$ for a duration of $\tau$, followed by a backward imaginary time evolution under the  Hamiltonian with $\gamma=1$ for the same duration. The resulting state then serves as the initially hotter state:
\begin{equation}
    \ket{\Psi^{\mathrm{h}}(0)} = \frac{e^{+H_1\,\tau_{\mathrm{pre}}/2}\, e^{-H_0\,\tau_{\mathrm{pre}}/2}\, \ket{\psi(\theta)}}{\| e^{+H_1\,\tau_{\mathrm{pre}}/2}\, e^{-H_0\,\tau_{\mathrm{pre}}/2}\, \ket{\psi(\theta)} \|},
\end{equation}
where $\tau_{\mathrm{pre}}$ is the pre-evolution duration. The state $ \ket{\Psi^{\mathrm{h}}(0)}$ is then evolved under $H_1$ from $\tau=0$. As shown in Fig.~\ref{Example}(a), the ratio $p_1/p_0$ at the end of the pre-evolution duration is lower than the ratio before the pre-evolution. We plot the evolving curves of modified average energy of the two initial states in Fig.~\ref{Example}(b), where $E_0$ is the ground state energy of $H_1$ obtained by exact diagonalization. Evidently, the initially hotter state prepared via pre-evolution, relaxes to the ground state faster than the original initial state.

\begin{figure*}
\begin{centering}
\includegraphics[scale=0.44]{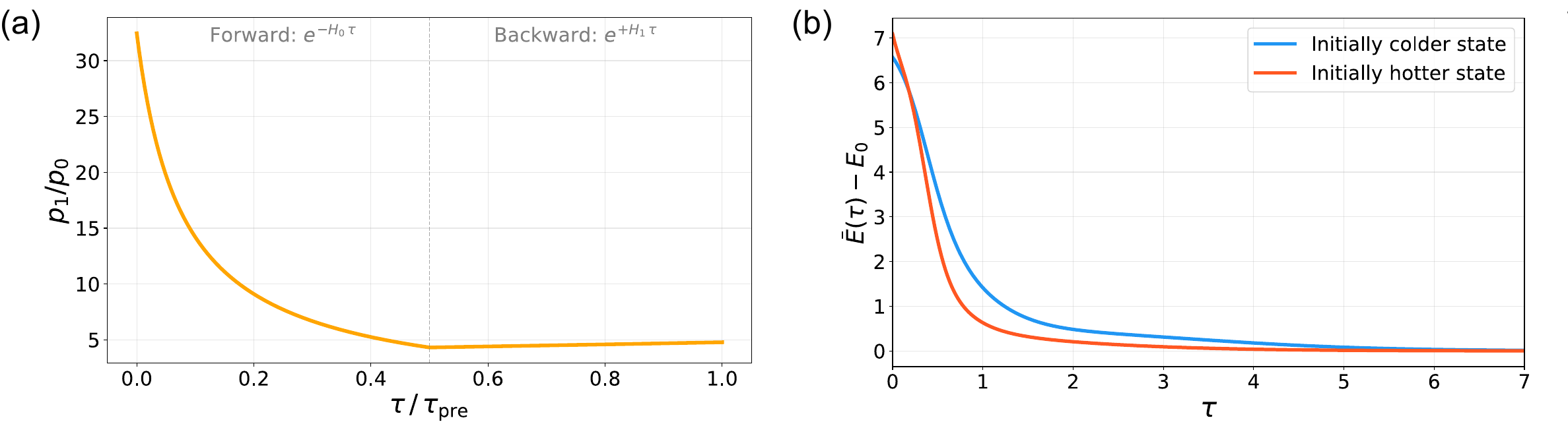}
\par\end{centering}
\caption{\textbf{(a):} The change of the ratio $p_1/p_0$ during the pre-evolution duration.
\textbf{(b):} The images of evolving curves of average energy. Here $L = 8$, $\mu = 0.3$, $\theta = 0.1\pi$, $\gamma_0=0.01$, $\gamma_1=1$, $\tau_{\mathrm{pre}}=0.2$.
\label{Example}}
\end{figure*}

\section{The proof of colinearity preservation}\label{colinearity}

\begin{theorem}
Assume that the initial \textcolor{black}{population} distributions of three states, denoted as $\vec{p}_A(0)$, $\vec{p}_B(0)$, and $\vec{p}_C(0)$ in an $n$-dimensional space, are collinear. Under imaginary time evolution governed by a given Hamiltonian $H$, their \textcolor{black}{population} distributions $\vec{p}_A(\tau)$, $\vec{p}_B(\tau)$, and $\vec{p}_C(\tau)$ remain collinear at any time $\tau$.
\end{theorem}

\begin{proof}
Let the initial population distributions be collinear. By definition, the ratio of the differences between their corresponding components $i \in \{0, 1, \dots, n-1\}$ is a constant $\lambda$ independent of $i$:
\begin{equation}
    \frac{p_{C,i}(0) - p_{A,i}(0)}{p_{B,i}(0) - p_{A,i}(0)} = \lambda
\end{equation}
This can be rewritten as:
\begin{equation}
    p_{C,i}(0) = \lambda p_{B,i}(0) + (1-\lambda) p_{A,i}(0)
\end{equation}

The normalized population distribution for state $k \in \{A, B, C\}$ at time $\tau$ is given by:
\begin{equation}
    p_{k,i}(\tau) = \frac{p_{k,i}(0) W_i(\tau)}{Z_k(\tau)}
\end{equation}
where $W_i(\tau) = e^{-2E_i \tau}$ and $Z_k(\tau) = \sum_{j} p_{k,j}(0) W_j(\tau)$. 

Using the linear relation at $\tau=0$, $Z_C(\tau)$ can be expressed as:
\begin{align}
    Z_C(\tau) &= \sum_{j} \left[ \lambda p_{B,j}(0) + (1-\lambda) p_{A,j}(0) \right] W_j(\tau) \\&= \lambda Z_B(\tau) + (1-\lambda) Z_A(\tau)\label{eq:co}
\end{align}

Now, we evaluate the normalized component $p_{C,i}(\tau)$:
\begin{align}
    p_{C,i}(\tau) &= \frac{\lambda p_{B,i}(0) W_i(\tau) + (1-\lambda) p_{A,i}(0) W_i(\tau)}{Z_C(\tau)} \nonumber \\
    &= \frac{\lambda Z_B(\tau) p_{B,i}(\tau) + (1-\lambda) Z_A(\tau) p_{A,i}(\tau)}{Z_C(\tau)}
\end{align}

Subtracting $p_{A,i}(\tau)$ from both sides yields the difference:
% \begin{widetext}
\begin{equation}
    p_{C,i}(\tau) - p_{A,i}(\tau) = \frac{\lambda Z_B(\tau) p_{B,i}(\tau) + (1-\lambda) Z_A(\tau) p_{A,i}(\tau) - Z_C(\tau) p_{A,i}(\tau)}{Z_C(\tau)}
\end{equation}
% \end{widetext}
Substituting Eq. (\ref{eq:co}) into the last term of the numerator, the terms containing $Z_A(\tau)$ cancel out perfectly:

\begin{align}
    p_{C,i}(\tau) - p_{A,i}(\tau) &= \frac{\lambda Z_B(\tau) p_{B,i}(\tau) - \lambda Z_B(\tau) p_{A,i}(\tau)}{Z_C(\tau)} \nonumber \\
    &= \frac{\lambda Z_B(\tau)}{Z_C(\tau)} \left[ p_{B,i}(\tau) - p_{A,i}(\tau) \right]
\end{align}

Dividing both sides by $p_{B,i}(\tau) - p_{A,i}(\tau)$, we obtain the ratio of the differences at time $\tau$:
\begin{equation}
    \frac{p_{C,i}(\tau) - p_{A,i}(\tau)}{p_{B,i}(\tau) - p_{A,i}(\tau)} = \frac{\lambda Z_B(\tau)}{Z_C(\tau)} \equiv \lambda'(\tau)
\end{equation}

Since the new ratio $\lambda'(\tau)$ depends only on the initial constant $\lambda$ and the global partition functions at time $\tau$, it is strictly independent of the index $i$. Therefore, the population distributions $\vec{p}_A(\tau)$, $\vec{p}_B(\tau)$, and $\vec{p}_C(\tau)$ remain collinear at any time $\tau > 0$.
\end{proof}

Therefore, if all initial states lie on a straight line, their population distributions will remain collinear at any $\tau$ during the evolution. Consequently, the shape of the isochronous lines for these states remains invariant (maintaining a straight line). Since the rate of change of $p_i(\tau)$ for an evolving state is dictated by $\frac{dp_{i}(\tau)}{d\tau} = -2 (E_i - \bar{E}(\tau)) p_{i}(\tau)$, the population distribution of each state changes in a different direction at the same imaginary time, causing the isochronous line to rotate within the population simplex over time. Furthermore, since all isoenergy lines are parallel to each other, there may be a specific time $\tau^*$ during this rotation when the isochronous line becomes parallel to the isoenergy lines and coincides with one of them.

\section{A sufficient condition for finite-time Mpemba effects for a general distance function}\label{general}

\begin{theorem}\label{thm: sufficient condition general F}
	For any measure $\FN$ defined as in the main text, the Mpemba effect occurs before the threshold $\epsilon$ if 
	\begin{align}
		\frac{\epsilon}{\fn(E_{1})} < \min \left \{\frac{a}{(1-\frac{f(E_{1})}{f(E_{n-1})})(1+a)}, \quad \left(1+\frac{1}{\rc} \left(\frac{\sh-\rh}{\frac{\fn(E_{1})}{(1+a)\fn(E_{n-1})}\rc - \rh}\right)^{\frac{E_{1}}{E_{2}-E_{1}}} \right)^{-1} \right \},
	\end{align}
	for any $a\in(0, \frac{\fn(E_{1})\rc}{\fn(E_{n-1})\rh} - 1)$, $\frac{\rc}{\rh} > \frac{\fn(E_{n-1})}{\fn(E_{1})}$, and $\fn(E_{n-1})>\fn(E_{1})$.

\end{theorem}

\begin{proof}
	Let us write the measures $\FN^{\mathrm{h}}$ and $\FN^{\mathrm{c}}$ explicitly, which read
	\begin{align}
		\FNh{\tau} &= \frac{\sum_{i=0}^{n-1}\fn(E_{i})\ph_{i}(0)e^{-2E_{i}\tau}}{\sum_{i=0}^{n-1}\ph_{i}(0)e^{-2E_{i}\tau}},\\
		\FNc{\tau} &= \frac{\sum_{i=0}^{n-1}\fn(E_{i})\pc_{i}(0)e^{-2E_{i}\tau}}{\sum_{i=0}^{n-1}\pc_{i}(0)e^{-2E_{i}\tau}}.
	\end{align}
	As we did in the proof of Theorem 2, we may define measures $\tilde{D}_f^{\mathrm{h}}$ and $\tilde{D}_f^{\mathrm{c}}$ such that 
	\begin{align}
		\tilde{D}_f^{\mathrm{h}}(\tau) \geq \FNh{\tau},\quad \tilde{D}_f^{\mathrm{c}}(\tau) \leq \FNc{\tau},\quad \forall \tau\geq0.
	\end{align}
	Then $\tilde{D}_f^{\mathrm{h}}(\tau)<\tilde{D}_f^{\mathrm{c}}(\tau)$ becomes the sufficient condition for $\FNh{\tau}<\FNc{\tau}$.
	Observe that 
	\begin{align}
		\fn(E_{n-1})(1-\ph_{0}(\tau)) &= \frac{\sum_{i=1}^{n-1}\fn(E_{n-1})\ph_{i}e^{-2E_{i}\tau}}{\sum_{i=0}^{n-1}\ph_{i}e^{-2E_{i}\tau}} \geq \FNh{\tau},\\
		\fn(E_{1})(1-\pc_{0}(\tau)) &= \frac{\sum_{i=1}^{n-1}\fn(E_{1})\pc_{i}e^{-2E_{i}\tau}}{\sum_{i=0}^{n-1}\pc_{i}e^{-2E_{i}\tau}} \leq \FNc{\tau}.
	\end{align}
	Furthermore, we have already defined $\pth_{0}(\tau)$ and $\ptc_{0}(\tau)$ in Eqs.~(17) and~(19) of main text so that 
	\begin{align}
		1-\pth_{0}(\tau) \geq 1-\ph_{0}(\tau),\quad 1-\ptc_{0}(\tau) \leq 1-\pc_{0}(\tau),
	\end{align}
	for all $\tau\geq0$.
	Hence, we choose 
	\begin{align}
		\tilde{D}_f^{\mathrm{h}}(\tau) &\coloneqq \fn(E_{n-1})(1-\pth_{0}(\tau)),\label{eq: Fth general def}\\
		\tilde{D}_f^{\mathrm{c}}(\tau) &\coloneqq \fn(E_{1})(1-\ptc_{0}(\tau)),\label{eq: Ftc general def}
	\end{align}
	where we can write
	\begin{align}
		1-\pth_{0}(\tau) &= \frac{\rh e^{-2E_{1}\tau} + (\sh-\rh)e^{-2E_{2}\tau}}{1 + \rh e^{-2E_{1}\tau} + (\sh-\rh)e^{-2E_{2}\tau}},\\
		1-\ptc_{0}(\tau) &= \frac{\rc e^{-2E_{1}\tau}}{1 + \rc e^{-2E_{1}\tau}},
	\end{align}
	using Eqs.~(20) and~(21) of main text.
	
		Let $\tilde{\tau}^{\star}$ be a positive number such that  $\tilde{D}_f^{\mathrm{h}}(\tau) <\tilde{D}_f^{\mathrm{c}}(\tau)$ for all $\tau>\tilde{\tau}^{\star}$.
		This is equivalent to
		\begin{align}
			\frac{\rc e^{-2E_{1}\tau}}{1+\rc e^{-2E_{1}\tau}} >\frac{\fn(E_{n-1})}{\fn(E_{1})} \frac{\rh e^{-2E_{1}\tau} + (\sh-\rh)e^{-2E_{2}\tau}}{1+\rh e^{-2E_{1}\tau} + (\sh-\rh)e^{-2E_{2}\tau}},
		\end{align}
		for all $\tau>\tilde{\tau}^{\star}$, or, after some algebra
		\begin{align}\label{eq: Fth Fch with ratio}
			\frac{\fn(E_{n-1})}{\fn(E_{1})} \left( \rh e^{-2E_{1}\tau} + (\sh-\rh)e^{-2E_{2}\tau} \right)< \frac{\rc e^{-2E_{1}\tau}}{1+(1-\frac{f(E_{1})}{f(E_{n-1})})\rc e^{-2E_{1}\tau}}. 
		\end{align}
		
		Now we can write a sufficient condition for Eq.~\eqref{eq: Fth Fch with ratio}:
		For any $a>0$ and $\fn(E_{n-1})>\fn(E_{1})$, if there exists $\tau^{\circledast}_{0}\geq\tilde{\tau}^{\star}$ that 
		\begin{align}\label{eq: Fth Fch with ratio 3}
			\left (1-\frac{f(E_{1})}{f(E_{n-1})}\right )\rc e^{-2E_{1}\tau} < a,
		\end{align}
		and 
		\begin{align}\label{eq: Fth Fch with ratio 2}
			\frac{(1+a)\fn(E_{n-1})}{\fn(E_{1})}\left(\rh + (\sh-\rh)e^{-2(E_{2}-E_{1})\tau}\right) < \rc,
		\end{align}
for all $\tau>\tau^{\circledast}_{0}$, then the inequality in Eq.~\eqref{eq: Fth Fch with ratio} is satisfied. 
		Eq.~\eqref{eq: Fth Fch with ratio 3} is equivalent to $\tau > \tau^{\circledast}_{1}$ with
		\begin{align}\label{eq: tau circledast 1}
			\tau^{\circledast}_{1} = \frac{1}{2E_{1}} \left [\ln\left (\frac{\rc}{a}\right ) + \ln\left( 1-\frac{f(E_{1})}{f(E_{n-1})}\right) \right ],
		\end{align}
		whereas Eq.~\eqref{eq: Fth Fch with ratio 2} is equivalent to $\tau > \tau^{\circledast}_{2}$ with
		\begin{align}\label{eq: tau circledast 2}
			\tau^{\circledast}_{2} = \frac{1}{2(E_{2}-E_{1})}\ln\left(\frac{\sh - \rh}{\frac{\fn(E_{1})}{(1+a)\fn(E_{n-1})}\rc - \rh}\right) .
		\end{align}
		For $\tau^{\circledast}_{2}$ to be a real number, we must have $\frac{\fn(E_{1})}{(1+a)\fn(E_{n-1})}\rc - \rh>0$, i.e. $a<\frac{\fn(E_{1})\rc}{\fn(E_{n-1})\rh} - 1$.
		Since $a>0$, our sufficient condition is only capable of detecting the Mpemba effect when $\frac{\fn(E_{1})\rc}{\fn(E_{n-1})\rh} - 1>0$, i.e. $\frac{\rc}{\rh} > \frac{\fn(E_{n-1})}{\fn(E_{1})}$.
		
		Next we prove that $\tau^{\circledast}_{0}=\max\{\tau^{\circledast}_{1}, \tau^{\circledast}_{2}\}\geq\tilde{\tau}^{\star}$. It is evident that $\tau^{\circledast}_{1}$ is a monotonically decreasing function of $a$, and $\tau^{\circledast}_{1}\to \infty$ when $a\to 0$, while $\tau^{\circledast}_{2}$ is a monotonically increasing function of $a$, and $\tau^{\circledast}_{2}\to \infty$ when $a\to \frac{\fn(E_{1})\rc}{\fn(E_{n-1})\rh} - 1$. As a result, there must exsit $a^*\in(0, \frac{\fn(E_{1})\rc}{\fn(E_{n-1})\rh} - 1)$ that satisfies $\tau^{\circledast}_{1}(a^*)=\tau^{\circledast}_{2}(a^*)$, and $\tau^{\circledast}_{0}(a^*)$ is the minimum of $\tau^{\circledast}_{0}(a)$.
From Eqs.~\eqref{eq: tau circledast 1} and \eqref{eq: tau circledast 2}, we can get 
       \begin{equation}
         \frac{\fn(E_{n-1})}{\fn(E_{1})} \left( \rh e^{-2E_{1}\tau^{\circledast}_{0}(a^*)} + (\sh-\rh)e^{-2E_{2}\tau^{\circledast}_{0}(a^*)} \right)= \frac{\rc e^{-2E_{1}\tau^{\circledast}_{0}(a^*)}}{1+(1-\frac{f(E_{1})}{f(E_{n-1})})\rc e^{-2E_{1}\tau}}, 
       \end{equation}
        which shows that $\tau^{\circledast}_{0}(a^*)$ is exactly $\tilde{\tau}^{\star}$. Therefore, we have proven that $\tau^{\circledast}_{0}(a)\geq\tau^{\circledast}_{0}(a^*)=\tilde{\tau}^{\star}$.

	Finally, our sufficient condition for finite-time Mpemba effects can be written as $\epsilon < \tilde{D}_f^{\mathrm{c}}(\tau^{\circledast}_{0})=\tilde{D}_f^{\mathrm{c}}(\max\{\tau^{\circledast}_{1},\tau^{\circledast}_{2}\})$.
	Since $\tilde{D}_f^{\mathrm{c}}$ is a monotonic function, this is equivalent to having $\epsilon < \min\{\tilde{D}_f^{\mathrm{c}}(\tau^{\circledast}_{1}), \tilde{D}_f^{\mathrm{c}}(\tau^{\circledast}_{2})\}$.
	Using Eqs.~\eqref{eq: Ftc general def}, \eqref{eq: tau circledast 1}, and \eqref{eq: tau circledast 2}, we obtain 
	\begin{align}
		\frac{\epsilon}{\fn(E_{1})} < \min \left \{\frac{a}{(1-\frac{f(E_{1})}{f(E_{n-1})})(1+a)}, \quad \left(1+\frac{1}{\rc} \left(\frac{\sh-\rh}{\frac{\fn(E_{1})}{(1+a)\fn(E_{n-1})}\rc - \rh}\right)^{\frac{E_{1}}{E_{2}-E_{1}}} \right)^{-1} \right \},
	\end{align}
    for any $a\in(0, \frac{\fn(E_{1})\rc}{\fn(E_{n-1})\rh} - 1)$, $\frac{\rc}{\rh} > \frac{\fn(E_{n-1})}{\fn(E_{1})}$, and $\fn(E_{n-1})>\fn(E_{1})$. Notably, for $a=a^*$, we get the optimal sufficient condition $\epsilon<\frac{a^*\fn(E_{1})}{(1-\frac{f(E_{1})}{f(E_{n-1})}((1+a^*)}$.
\end{proof}

\section{Proof of Theorem 1 when the first excited state population is zero}\label{app: first excited zero}

To prove the necessity, we first suppose that $\rh > \rc$. 
Straightforwardly, if $\rc=0$, we just need to change the part for $\mathrm{x}=\mathrm{c}$ as:
\begin{equation}
		\FNc{\tau} 
		= \frac{r_k^{\mathrm{c}}\fn(E_{k})e^{-2E_{k}\tau}+\eta^{\mathrm{c}}(\tau)e^{-2E_{k+1}\tau}}{1+r_k^{\mathrm{c}}e^{-2E_{k}\tau} + o^{\mathrm{c}}(\tau)e^{-2E_{k+1}\tau}},
	\end{equation}
	where
	\begin{align}
		o^{\mathrm{c}}(\tau) &=  \sum_{i=k+1}^{n-1}\frac{p^{\mathrm{c}}_{i}(0)}{p^{\mathrm{c}}_{0}(0)}e^{-2(E_{i}-E_{k+1})\tau},\\
		\eta^{\mathrm{c}}(\tau) &= \sum_{i=k+1}^{n-1}\frac{p^{\mathrm{c}}_{i}(0)\fn(E_{i})e^{-2(E_{i}-E_{k+1})\tau}}{p^{\mathrm{c}}_{0}(0)}.
	\end{align} 
	Here $\rc_k = \frac{\pc_{k}(0)}{\pc_{0}(0)}$, and $k$ is the smallest index of excited state with nonzero population.
    
    Now we rewrite 
	\begin{align}
		\FNc{\tau} = R^{\mathrm{c}}(\tau)e^{-2E_{k}\tau} + \mu^{\mathrm{c}}(\tau)e^{-2E_{k}\tau},
	\end{align}
	where
	\begin{align}
		R^{\mathrm{c}}(\tau) &= \frac{r_k^{\mathrm{c}}\fn(E_{k})}{1+r_k^{\mathrm{c}}e^{-2E_{k}\tau}},\\
		\mu^{\mathrm{c}}(\tau) &= \frac{\eta^{\mathrm{c}}(\tau) - R^{\mathrm{c}}(\tau)o^{\mathrm{c}}(\tau)e^{-2E_{k}\tau}}{1+r^{\mathrm{c}}e^{-2E_{k}\tau} + o^{\mathrm{c}}(\tau)e^{-2E_{k+1}\tau}}.
	\end{align}
	Then, for any $\varepsilon_{1}>0$, there exists $\tau_{1}$ such that 
	\begin{align}
		R^{\mathrm{c}}(\tau)(1- \varepsilon_{1}) < \FNc{\tau} e^{2E_{k}\tau} < R^{\mathrm{c}}(\tau) (1+ \varepsilon_{1}),
	\end{align}
	for all $\tau>\tau_{1}$, implies that 
	\begin{align}
		\frac{R_1^{\mathrm{h}}(\tau)}{R_k^{\mathrm{c}}(\tau)}\frac{1-\varepsilon_{1}}{1+\varepsilon_{1}}<\frac{\FNh{\tau}e^{2E_{1}\tau}}{\FNc{\tau}e^{2E_{k}\tau}}<\frac{R_1^{\mathrm{h}}(\tau)}{R_k^{\mathrm{c}}(\tau)}\frac{1+\varepsilon_{1}}{1-\varepsilon_{1}}.
	\end{align}
	Similarly, for any $\varepsilon_{2}>0$, there exists $\tau_{2}$ such that
	\begin{align}
		\frac{\rh_1 \fn(E_1)}{\rc_k\fn(E_2)}\frac{1-\varepsilon_{2}}{1+\varepsilon_{2}}<\frac{R^{\mathrm{h}}(\tau)}{R^{\mathrm{c}}(\tau)}<\frac{\rh_1\fn(E_1)}{\rc_k\fn(E_k)}\frac{1+\varepsilon_{2}}{1-\varepsilon_{2}}
	\end{align}
	for all $\tau > \tau_{2}$.
	Combined, we arrive at the statement that for any $\varepsilon_{3}>0$, there exists $\tau_{3}$ such that
	\begin{align}
		\frac{\rh_1 \fn(E_1)}{\rc_k\fn(E_2)}\frac{1-\varepsilon_{3}}{1+\varepsilon_{3}}<\frac{\FNh{\tau}e^{2E_{1}\tau}}{\FNc{\tau}e^{2E_{k}\tau}}<\frac{\rh_1 \fn(E_1)}{\rc_k\fn(E_2)}\frac{1+\varepsilon_{3}}{1-\varepsilon_{3}}.
	\end{align}
there exists $\tau_{4}>\tau_{3}$ that for all $\tau>\tau_{4}$
 \begin{align}
		e^{2(E_{k}-E_1)\tau}\frac{\rh_1 \fn(E_1)}{\rc_k\fn(E_2)}\frac{1-\varepsilon_{3}}{1+\varepsilon_{3}}>1
	\end{align}
then we get
    \begin{align}
		1<e^{2(E_{k}-E_1)\tau}\frac{\rh_1 \fn(E_1)}{\rc_k\fn(E_2)}\frac{1-\varepsilon_{3}}{1+\varepsilon_{3}}<\frac{\FNh{\tau}}{\FNc{\tau}}
	\end{align}
	we can prove that $\FNh{\tau}>\FNc{\tau}$ for all $\tau>\tau_{4}$. 
    The proof of sufficiency is similar.

\section{The crossing time in the long-time regime}
In practical computation, we stop the evolution once the distance function of an evolution state falls below a certain threshold (i.e., $D_f(\tau)\le\epsilon$), regarding it as having reached the ground state. Therefore, if the time at which the Mpemba effect occurs is longer than the evolution time, we will not be able to observe the Mpemba effect. Thus, we need to know the time when the Mpemba effect takes place, namely the crossing time $\tau^*$ of the distance function of the two evolution states. When the crossing time is evidently shorter than the stopping time, i.e., Mpemba effects occur in the short‑ or intermediate-time regime, we do not need to worry about this issue. What we truly need to focus on are cases where the crossing time is comparable to or much larger than the stopping time, i.e., Mpemba effects occur in the long-time regime. Take $D_f(\tau)=\bar{E}(\tau)$ as an example, since after the stopping time $\tau_s$ the average energy of the evolution state has already fallen below the threshold, it implies that average energies of both evolution states satisfy $\bar{E}(\tau^*)\approx0$ around the crossing time in such cases. Thus the crossing time can be derived by following steps.

\textbf{Case I: $k^{\text{h}}_1=k^{\text{c}}_1=k_1$, and $\frac{p^{\text{h}}_{k_1}(0)}{p^{\text{h}}_{0}(0)}<\frac{p^{\text{c}}_{k_1}(0)}{p^{\text{c}}_{0}(0)}$ (If $p^{\text{h}}_{k_1}=p^{\text{c}}_{k_1}$, then compare higher energy levels with nonzero occupancy numbers).} For this case, the average energy at the crossing time $\tau^*$ can be approximated as 
\begin{align}
    \bar{E}^{\text{h}}(\tau^*)&\approx p^{\text{h}}_{k_1}(\tau^*)E_{k_1}+p^{\text{h}}_{k^{\text{h}}_2}(\tau^*)E_{k^{\text{h}}_2},\label{eq:c1} \\ 
     \bar{E}^{\text{c}}(\tau^*)&\approx p^{\text{c}}_{k_1}(\tau^*)E_{k_1},\label{eq:c2}
\end{align}
here $k^{\text{h/c}}_{1}$ and $k^{\text{h}}_{2}$ are the lowest and second lowest excited energy levels with nonzero occupation in the corresponding initial distribution. The reason that Eq. (\ref{eq:c1}) contains one more term than Eq. (\ref{eq:c2}) is that, when the Mpemba effect exists, the first term of Eq. (\ref{eq:c1}) is necessarily smaller than that of Eq. (\ref{eq:c2}). Therefore, the additional term in Eq. (\ref{eq:c1}) is required to make the two equations equal.

Then we get
\begin{align}
\bar{E}^{\text{h}}(\tau^*)=
\bar{E}^{\text{c}}(\tau^*)
\Rightarrow &\frac{\frac{p^{\text{h}}_{k_1}(0)}{p^{\text{h}}_{0}(0)}e^{-2E_{k_1}\tau^*}E_{k_1}+\frac{p^{\text{h}}_{k_2^{\text{h}}}(0)}{p^{\text{h}}_{0}(0)}e^{-2E_{k_2^{\text{h}}}\tau^*}E_{k_2^{\text{h}}}}{1+\sum_{i>0}\frac{p^{\text{h}}_{i}(0)}{p^{\text{h}}_{0}(0)}e^{-2E_i\tau^*}}\approx\frac{\frac{p^{\text{c}}_{k_1}(0)}{p^{\text{c}}_{0}(0)}e^{-2E_{k_1}\tau^*}E_{k_1}}{1+\sum_{i>0}\frac{p^{\text{c}}_i(0)}{p^{\text{c}}_{0}(0)}e^{-2E_i\tau^*}}\\
\Rightarrow &\frac{p^{\text{h}}_{k_1}(0)}{p^{\text{h}}_{0}(0)}e^{-2E_{k_1}\tau^*}E_{k_1}+\frac{p^{\text{h}}_{k_2^{\text{h}}}(0)}{p^{\text{h}}_{0}(0)}e^{-2E_{k_2^{\text{h}}}\tau^*}E_{k_2^{\text{h}}}\approx\frac{p^{\text{c}}_{k_1}(0)}{p^{\text{c}}_{0}(0)}e^{-2E_{k_1}\tau^*}E_{k_1}\\
\Rightarrow &
\tau^*\approx\frac{1}{2(E_{k^{\text{h}}_2}-E_{k_1})}\ln{\frac{\frac{p^{\text{h}}_{k^{\text{h}}_2}(0)}{p^{\text{h}}_{0}(0)}E_{k^{\text{h}}_2}}{(\frac{p^{\text{c}}_{k_1}(0)}{p^{\text{c}}_{0}(0)}-\frac{p^{\text{h}}_{k_1}(0)}{p^{\text{h}}_{0}(0)})E_{k_1}}}.\label{eq:tau1}
\end{align}
So if $\tau^*$ of Eq. (\ref{eq:tau1}) is shorter than stopping time, then we can observe the Mpemba effect, and vice versa. Note that $\frac{p^{\text{h}}_{k^{\text{h}}_2}(0)}{p^{\text{h}}_{0}(0)}E_{k^{\text{h}}_2}$ should be greater than $(\frac{p^{\text{c}}_{k_1}(0)}{p^{\text{c}}_{0}(0)}-\frac{p^{\text{h}}_{k_1}(0)}{p^{\text{h}}_{0}(0)})E_{k_1}$, which is usually sufficient. Otherwise, the crossing time would be much shorter than the stop time.

\textbf{Case II: $k^{\text{h}}_1\ne k^{\text{c}}_1$.} For this case, the average energy at the crossing time $\tau^*$ can be approximated as 
\begin{align}
    \bar{E}^{\text{h}}(\tau^*)&\approx p^{\text{h}}_{k^{\text{h}}_1}(\tau^*)E_{k^{\text{h}}_1},\label{eq:c9} \\ 
     \bar{E}^{\text{c}}(\tau^*)&\approx p^{\text{c}}_{k^{\text{c}}_1}(\tau^*)E_{k^{\text{c}}_1},\label{eq:c10}
\end{align}
In this situation, Eq. (\ref{eq:c9}) contains only one term because that the occurrence of the Mpemba effect requires $k^{\text{h}}_1> k^{\text{c}}_1$, and the first term of Eq. (\ref{eq:c9}) can be equal to that of Eq. (\ref{eq:c10}). Therefore, we get
\begin{align}
\bar{E}^{\text{h}}(\tau^*)=
\bar{E}^{\text{c}}(\tau^*)
\Rightarrow &\frac{\frac{p^{\text{h}}_{k_1^{\text{h}}}(0)}{p^{\text{h}}_{0}(0)}e^{-2E_{k_1^{\text{h}}}\tau^*}E_{k_1^{\text{h}}}}{1+\sum_{i>0}\frac{p^{\text{h}}_{i}(0)}{p^{\text{h}}_{0}(0)}e^{-2E_i\tau^*}}\approx\frac{\frac{p^{\text{c}}_{k_1^{\text{c}}}(0)}{p^{\text{c}}_{0}(0)}e^{-2E_{k_1^{\text{c}}}\tau^*}E_{k_1^{\text{c}}}}{1+\sum_{i>0}\frac{p^{\text{c}}_i(0)}{p^{\text{c}}_{0}(0)}e^{-2E_i\tau^*}}\\
\Rightarrow &\frac{p^{\text{h}}_{k_1^{\text{h}}}(0)}{p^{\text{h}}_{0}(0)}e^{-2E_{k_1^{\text{h}}}\tau^*}E_{k_1^{\text{h}}}\approx\frac{p^{\text{c}}_{k_1^{\text{c}}}(0)}{p^{\text{c}}_{0}(0)}e^{-2E_{k_1^{\text{c}}}\tau^*}E_{k_1^{\text{c}}}\\
\Rightarrow &
\tau^*\approx\frac{1}{2(E_{k^{\text{h}}_1}-E_{k^{\text{c}}_1})}\ln{\frac{p^{\text{h}}_{k^{\text{h}}_1}(0)p^{\text{c}}_{0}(0)E_{k^{\text{h}}_1}}{p^{\text{c}}_{k^{\text{c}}_1}(0)p^{\text{h}}_{0}(0)E_{k^{\text{c}}_1}}}.\label{eq:tau2}
\end{align}
So if $\tau^*$ of Eq. (\ref{eq:tau2}) is shorter than stopping time, then we can observe the Mpemba effect, and vice versa. Note that $\frac{p^{\text{h}}_{k^{\text{h}}_1}(0)}{p^{\text{h}}_{0}(0)}E_{k^{\text{h}}_1}$ should be greater than $\frac{p^{\text{c}}_{k^{\text{c}}_1}(0)}{p^{\text{c}}_{0}(0)}E_{k^{\text{c}}_1}$, which is usually sufficient. Otherwise, the crossing time would be much shorter than the stop time. 

The numercal results in Fig.~\ref{Appendix} shows that our estimated crossing time (gray dotted lines) of Eq. (\ref{eq:tau1}) and Eq. (\ref{eq:tau2}) are very close to the actual crossing time (purple dashed lines).

\begin{figure*}
\begin{centering}
\includegraphics[scale=0.55]{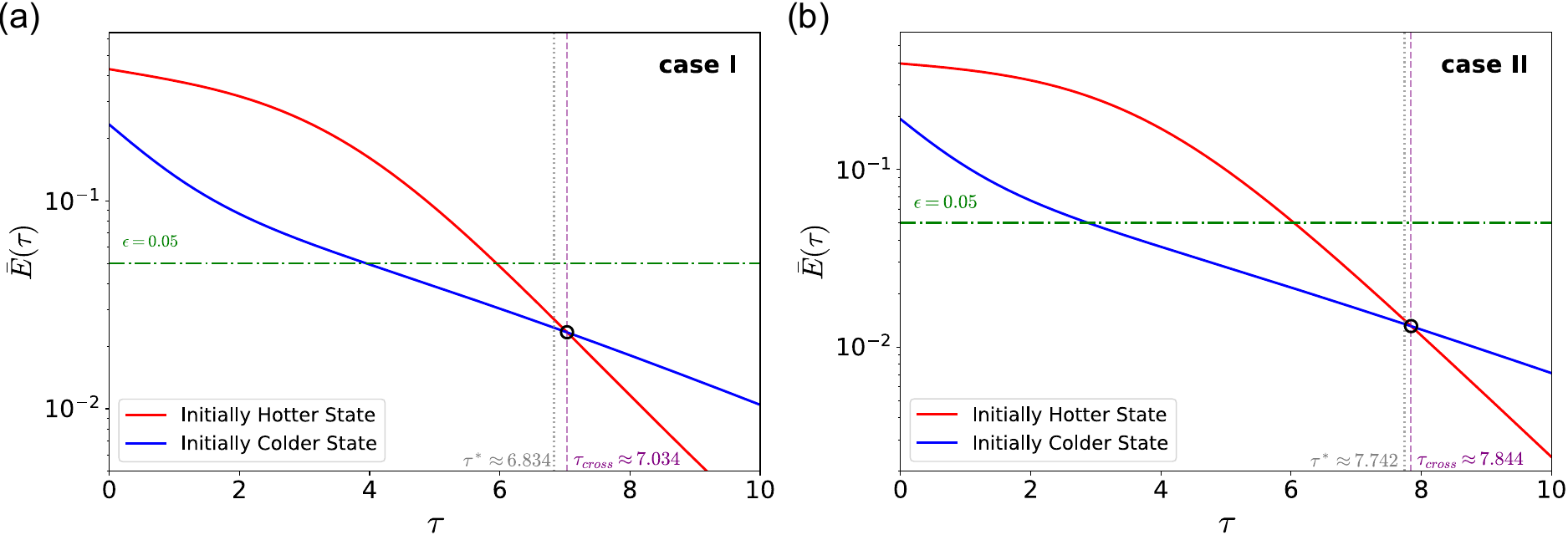}
\par\end{centering}
\caption{The images of evolving curves of average energy. \textbf{(a):} Initially hotter state: $(0.05, 0.005, 0.8, 0.05, 0.095)$. Initially colder state: $(0.3, 0.45, 0.05, 0.1, 0.1)$. The black circle denotes the crossing time $\tau \approx7.034$. The gray dotted line indicate the estimated crossing time of Eq. (\ref{eq:tau1}).
\textbf{(b):} Initially hotter state: $(0.05, 0.0, 0.9, 0.025, 0.025)$. Initially colder state: $(0.4, 0.4, 0.05, 0.05, 0.1)$. The black circle denotes the crossing time $\tau \approx7.844$. The gray dotted line indicate the estimated crossing time of Eq. (\ref{eq:tau2}).
 Here $E_0=0, E_1=0.15, E_2=0.4, E_3=0.65, E_4=0.8$.
\label{Appendix}}
\end{figure*}